\documentclass{article}

\PassOptionsToPackage{numbers, compress}{natbib}
 \usepackage[final]{nips_2016}

\usepackage[utf8]{inputenc} 
\usepackage[T1]{fontenc}    
\usepackage{url}            
\usepackage{booktabs}       
\usepackage{amsfonts}       
\usepackage{nicefrac}       
\usepackage{microtype}      

\usepackage{url}
\usepackage{amsmath}
\usepackage{graphicx} 
\usepackage{subfigure} 

\usepackage{algorithm}
\usepackage{algorithmic}
\usepackage{amsmath}[11pt]
\usepackage{amsmath, amsthm, amssymb}
\usepackage{amsthm}
\usepackage{amsfonts}[11pt]
\usepackage{mathtools}
\usepackage{bbm}

\DeclareMathOperator*{\argmax}{argmax}

\newtheorem{lemma}{Lemma}
\newtheorem{theorem}[lemma]{Theorem}

\newcommand{\E}{\mathbb{E}}

\usepackage[inline]{enumitem} 

\title{Phased Exploration with Greedy Exploitation in Stochastic Combinatorial Partial Monitoring Games}

\author{
  Sougata Chaudhur\\
  Department of Statistics\\
  University of Michigan Ann Arbor\\
   \texttt{sougata@umich.edu} \\
   \And
   Ambuj Tewari \\
   Department of Statistics and EECS\\
   University of Michigan Ann Arbor\\
   \texttt{tewaria@umich.edu}\\
}

\begin{document}

\maketitle

\begin{abstract}
Partial monitoring games are repeated games where the learner receives feedback that might be different from adversary's move or even the reward gained by the learner. Recently, a general model of combinatorial partial monitoring (CPM) games was proposed \cite{lincombinatorial2014}, where the learner's action space can be exponentially large and adversary samples its moves from a bounded, continuous space, according to a fixed distribution. The paper gave a confidence bound based algorithm (GCB) that achieves $O(T^{2/3}\log T)$ distribution independent and $O(\log T)$ distribution dependent regret bounds. The implementation of their algorithm depends on two separate offline oracles and the distribution dependent regret additionally requires existence of a unique optimal action for the learner. Adopting their CPM model, our first contribution is a Phased Exploration with Greedy Exploitation (PEGE) algorithmic framework for the problem. Different algorithms within the framework achieve $O(T^{2/3}\sqrt{\log T})$ distribution independent and $O(\log^2 T)$ distribution dependent regret respectively. Crucially, our framework needs only the simpler ``argmax'' oracle from GCB and the distribution dependent regret does not require existence of a unique optimal action. Our second contribution is another algorithm, PEGE2, which combines gap estimation with a PEGE algorithm, to achieve an $O(\log T)$ regret bound, matching the GCB guarantee but removing the dependence on size of  the learner's action space. However, like GCB, PEGE2 requires access to both offline oracles and the existence of a unique optimal action. Finally, we discuss how our algorithm can be efficiently applied to a CPM problem of practical interest: namely, online ranking with feedback at the top.
\end{abstract}

\section{Introduction}
\label{introduction}
Partial monitoring (PM) games are repeated games played between a learner and an adversary over discrete time points. At every time point, the learner and adversary each simultaneously select an action, from their respective action sets, and the learner gains a reward, which is a function of the two actions. In PM games, the learner receives limited feedback, which might neither be adversary's move (full information games) nor the reward gained (bandit games). In \emph{stochastic} PM games, adversary generates actions which are independent and identically distributed according to a distribution fixed before the start of the game and unknown to the learner. The learner's objective is to develop a learning strategy that incurs low regret over time, based on the feedback received during the course of the game. Regret is defined as the difference between cumulative reward of the learner's strategy and the best fixed learner's action in hindsight. The usual learning strategies in online games combine some form of exploration (getting feedback on certain learner's actions) and exploitation (playing the perceived optimal action based on current estimates).

There has been a substantial advance in our understanding of \emph{finite} PM games. Starting with early work in the 2000s \cite{piccolboni2001discrete,cesa2006}, this body of research reached a culmination point with a comprehensive and complete classification of finite PM games \citep{bartok2013}. We refer the reader to these works for more references and also note that newer results continue to appear \citep{komiyama2015regret}. Finite PM games restrict both the learner's and adversary's action spaces to be finite, with a very general feedback model. All finite partial monitoring games can be classified into one of four categories, with minimax regret $\Theta(T)$, $\Theta(T^{2/3})$, $\Theta(T^{1/2})$ and $\Theta(1)$. The classification is governed by \emph{global} and \emph{local} \emph{observability} properties pertaining to a game \citep{bartok2013}. Another line of work has extended traditional multi-armed bandit problem (MAB) \citep{auer2002finite} to include combinatorial action spaces for learner (CMAB) \citep{chen2013combinatorial,kveton2015tight}. The combinatorial action space can be exponentially large, rendering traditional MAB algorithms designed for small finite action spaces, impractical with regret bounds scaling with size of action space. The CMAB algorithms exploit a finite subset of base actions, which are specific to the structure of problem at hand, leading to practical algorithms and regret bounds that do not scale with, or scale very mildly with, the size of the learner's action space.   

While finite PM and CMAB problems have witnessed a lot of activity, there is only one paper \cite{lincombinatorial2014} on combinatorial partial monitoring (CPM) games, to the best of our knowledge. There the authors combined the combinatorial aspect of CMAB with the limited feedback aspect of finite PM games, to develop a CPM model. The model extended PM games to include combinatorial action spaces for learner, which might be exponentially large, and infinite action spaces for the adversary. Neither of these situations can be handled by generic algorithms for finite PM games. Specifically, the model considered an action space $\mathcal{X}$ for the learner, that has a small subset of actions defining a \emph{global observable set} (see Assumption 2 in Section~\ref{prelim}). The adversary's action space is a continuous, bounded vector space with the adversary sampling moves from a fixed distribution over the vector space. The reward function considered is a general non-linear function of learner's and adversary's actions, with some restrictions (see Assumptions 1 \& 3 in Section~\ref{prelim}). The model incorporated a linear feedback mechanism where the feedback received is a linear transformation of adversary's move. Inspired by the classic confidence bound algorithms for MABs, such as UCB \citep{auer2002finite}, the authors proposed a Global Confidence Bound (GCB) algorithm that enjoyed two types of regret bound. The first one was a distribution independent $O(T^{2/3}\log T)$ regret bound and the second one was a distribution dependent $O(\log T)$ regret bound. A distribution dependent regret bound involves factors specific to the adversary's fixed distribution, while distribution independent means the regret bound holds over all possible distributions in a broad class of distributions. Both bounds also had a logarithmic dependence on $|\mathcal{X}|$. The algorithm combined online estimation with two offline computational oracles. The first oracle finds the action(s) achieving maximum value of reward function over $\mathcal{X}$, for a particular adversary action (argmax oracle), and the second oracle finds the action(s) achieving second maximum value of reward function over $\mathcal{X}$, for a particular adversary action (arg-secondmax oracle).  Moreover, the distribution dependent regret bound requires existence of a \emph{unique} optimal learner action. The inspiration for the CPM model came from various applications like crowdsourcing and matching problems like matching products with customers.  

{\bf Our Contributions.} We adopt the CPM model proposed  earlier \citep{lincombinatorial2014}. However, instead of using upper confidence bound techniques, our work is motivated by another classic technique developed for MABs, namely that of forced exploration. This technique was already used in the classic paper of Robbins \citep{robbins1985some} and has also been called ``forcing with certainty equivalence'' in the control theory literature \citep{agrawal1989certainty}. We develop a Phased Exploration with Greedy Exploitation (PEGE) algorithmic framework (Section~\ref{algorithmanalysis}) borrowing the PEGE terminology from work on linearly parameterized bandits \cite{rusmevichientong2010linearly}. When the framework is instantiated with different parameters, it achieves $O(T^{2/3}\sqrt{\log T})$ distribution independent and $O(\log^2 T)$ distribution dependent regret. Significantly, the framework combines online estimation with only the argmax oracle from GCB, which is a practical advantage over requiring an additional arg-secondmax oracle. Moreover, the distribution dependent regret does not require existence of unique optimal action. Uniqueness of optimal action can be an unreasonable assumption, especially in the presence of a combinatorial action space. Our second contribution is another algorithm PEGE2 (Section~\ref{gap}) that combines a PEGE algorithm with Gap Estimation, to achieve a distribution dependent $O(\log T)$ regret bound, thus matching the GCB regret guarantee in terms of $T$ and gap. Here, gap refers to the difference between expected reward of optimal and second optimal learner's actions. However, PEGE2 does require access to both the oracles and existence of unique optimal action. On the other hand, like GCB, PEGE2 has the property that its regret is never larger than $O(T^{2/3} \log T)$ even when there is no unique optimal action. Another crucial fact is that all our regret bounds are independent of $|\mathcal{X}|$, only depending on the size of the small \emph{global observable set}. Thus, though we have adopted the CPM model \cite{lincombinatorial2014}, our regret bounds are meaningful for countably infinite or even continuous learner's action space, whereas GCB regret bound has an explicit logarithmic dependence on $|\mathcal{X}|$. We provide a detailed comparison of our work with the GCB algorithm in Section~\ref{comparison}. Finally, we discuss how our algorithms can be efficiently applied in the CPM problem of online ranking with feedback restricted to top ranked items (Section~\ref{application}), a setting already considered \cite{chaudhuri2015} but analyzed in a non-stochastic setting.

\section{Preliminaries and Assumptions}
\label{prelim}
The online game is played between a learner and an adversary, over discrete rounds indexed by $t=1,2,\ldots$. The learner's action set is denoted as $\mathcal{X}$ which can be exponentially large. The adversary's action set is the infinite set $[0,1]^n$. The adversary fixes a distribution $p$ on $[0,1]^n$ before start of the game (adversary's strategy), with $p$ unknown to the learner. At each round of the game, adversary samples $\theta(t) \in [0,1]^n$ according to $p$, with $\E_{\theta (t) \sim p} [\theta (t)]= \theta^*_p$. The learner chooses $x(t) \in \mathcal{X}$ and gets  reward $r(x(t),\theta(t))$. However, the learner might not get to know either $\theta(t)$ (as in a full information game) or $r(x(t),\theta(t))$ (as in a bandit game). In fact, the learner receives, as feedback, a linear transformation of $\theta(t)$.That is, every action $x \in \mathcal{X}$ has an associated transformation matrix $M_x \in\mathbb{R}^{m_x \times n}$. On playing action $x(t)$, the learner receives a feedback $M_{x(t)} \cdot \theta(t) \in \mathbb{R}^{m_x}$. Note that the game with the defined feedback mechanism subsumes full information and bandit games. $M_x= \mathbb{I}^{n \times n}, \ \forall x$ makes it a full information game since $M_x \cdot \theta = \theta$. If $r(x,\theta)= x \cdot \theta$, then $M_x= x \in \mathbb{R}^n$ makes it a bandit game. The dimension $n$, action space $\mathcal{X}$, reward function $r(\cdot,\cdot)$ and transformation matrices $M_x$, $\forall x \in \mathcal{X}$ are known to the learner. The goal of the learner is to minimize the expected regret, which, for a given time horizon $T$, is:
\begin{equation}
R(T)= T \cdot \max_{x \in \mathcal{X}} \E [r(x,\theta)] - \sum_{t=1}^T \E [r(x(t), \theta(t))]
\end{equation}
where the expectation in the first term is taken over $\theta$, w.r.t. distribution $p$, and the second expectation is taken over $\theta$ and possible randomness in the learner's algorithm.

For distribution dependent regret bounds, we define gaps in expected rewards:
Let $x^* \in S(x) = \argmax_{x \in X} \bar{r}(x, \theta^*_p)$. Then $\Delta_x = \bar{r}(x^*, \theta^*_p) - \bar{r}(x, \theta^*_p)$ , $\Delta_{max}= \max \{\Delta_x: x \in \mathcal{X}\}$ and $\Delta= \min\{\Delta_x: x \in \mathcal{X}, \Delta_x >0\}$.

{\bf Assumption 1. (Restriction on Reward Function)} The first assumption is that $\E_{ \theta \sim p} [r(x,\theta)]= \bar{r}(x,\theta^*_p)$, for some function $\bar{r}(\cdot,\cdot)$. That is, the expected reward is a function of $x$ and $\theta^*_p$, which is always satisfied if $r(x,\theta)$ is a linear function of $\theta$, or if  distribution $p$ happens to be any distribution with support $[0,1]^n$ and fully parameterized by its mean $\theta^*_p$. With this assumption, the expected regret becomes:
\begin{equation}
R(T)= T  \cdot \bar{r}(x^*,\theta^*_p) - \sum_{t=1}^T \E [\bar{r}(x(t), \theta^*_p)] .
\end{equation}

{\bf Assumption 2. (Existence of Global Observable Set)} The second assumption is on the existence of a \emph{global observable set}, which is a subset of learner's action set and is required for estimating an adversary's move $\theta$. The \emph{global observable set} is defined as follows: for a set of actions $\sigma= \{x_1,x_2, \ldots, x_{|\sigma|}\}$ $\subseteq \mathcal{X}$, let their transformation matrices be stacked in a top down fashion to obtain a $\mathbb{R}^{\sum_{i=1}^{|\sigma|} m_{x_i} \times n}$ dimensional matrix $M_{\sigma}$. $\sigma$ is said to be a global observable set if $M_{\sigma}$ has full column rank, i.e., rank($M_{\sigma}$) = $n$. Then, the Moore-Penrose pseudoinverse $M^{+}_{\sigma}$ satisfies $M^{+}_{\sigma} M_{\sigma} = \mathbb{I}^{n \times n}$. Without the assumption on the existence of global observable set, it might be the case that even if the learner plays all actions in $\mathcal{X}$ on same $\theta$, the learner might not be able to recover $\theta$ (as $M^{+}_{\sigma} M_{\sigma}= \mathcal{I}^{n \times n}$ will not hold without full rank assumption). In that case, learner might not be able to distinguish between $\theta^*_{p_1}$ and $\theta^*_{p_2}$, corresponding to two different adversary's strategies. Then, with non-zero probability, the learner can suffer $\Omega(T)$ regret and no learner strategy can guarantee a sub-linear in $T$ regret (the intuition forms the base of the \emph{global observability condition} in \cite{piccolboni2001discrete}). Note that the size of the global observable set is small, i.e., $|\sigma| \le n$. A global observable set can be found by including an action $x$ in $\sigma$ if it strictly increases the rank of $M_{\sigma}$, till the rank reaches $n$. There can, of course, be more than one global observable set.

{\bf Assumption 3. (Lipschitz Continuity of Expected Reward Function)} The third assumption is on the Lipschitz continuity of expected reward function in its second argument. More precisely, it is assumed that $\exists \ R>0$ such that $\forall \ x \in \mathcal{X}$, for any $\theta_1$ and $\theta_2$, $|\bar{r}(x,\theta_1) - \bar{r}(x,\theta_2)| \le R \|\theta_1 - \theta_2\|_2$. This assumption is reasonable since otherwise, a small error in estimation of mean reward vector $\theta^*_p$ can introduce a large change in expected reward, leading to difficulty in controlling regret over time. The Lipschitz condition holds trivially for expected reward functions which are linear in second argument. The continuity assumption, along with the fact that adversary's moves are in $[0,1]^n$, implies boundedness of expected reward for any learner's action and any adversary's action. We denote $R_{max} = \max_{ x \in \mathcal{X}, \theta \in [0,1]^n} \bar{r}(x,\theta)$.

The three assumptions above will be made throughout. However, the fourth assumption will only be made in a subset of our results.

{\bf Assumption 4. (Unique Optimal Action)} The optimal action $x^* = \argmax_{x \in \mathcal{X}} \bar{r}(x, \theta^*_p)$ is unique. Denote a second best action (which may not be unique)
by $x^{*}_{-} = \argmax_{x \in \mathcal{X},x \neq x^*} \bar{r}(x, \theta^*_p)$. Note that $\Delta = \bar{r}(x^*, \theta^*_p) - \bar{r}(x^{*}_{-}, \theta^*_p)$.

\section{Phased Exploration with Greedy Exploitation}
\label{algorithmanalysis}

Algorithm~\ref{alg:PhaseForced} (PEGE) uses the classic idea of doing exploration in phases that are successively further apart from each other. In between exploration phases, we select action greedily by completely trusting
the current estimates. The constant $\beta$ controls how much we explore in a given phase and the constant $\alpha$ along with the function $C(\cdot)$ determines how much we exploit.
This idea is classic in the bandit literature \cite{robbins1985some,agrawal1989certainty,rusmevichientong2010linearly} but has not been applied to the CPM framework to the best of our knowledge.

\floatstyle{ruled}
\newfloat{algorithm}{htbp}{loa}
\floatname{algorithm}{Algorithm}
\begin{algorithm}
\caption{The PEGE Algorithmic Framework}
\label{alg:PhaseForced}
\begin{tabbing}
tabs \= tabs \= tabs \= tabs \= tabs \kill
1: Inputs: $\alpha$, $\beta$ and function $C(\cdot)$ (to determine amount of exploration/exploitation in each phase).\\
\\
2: For  \=$b = 1, 2, \ldots,$\\
3: \> {\bf Exploration}\\
4: \> \> For $i = 1$ to $|\sigma|$ \\
5: \>\> \> For $j = 1$ to $b^{\beta}$  \\
6: \>\>\>\> Let $t_{j,i} = t$ and $\theta(t_{j,i},b)= \theta(t)$ where $t$ is current time point\\
7: \>\>\>\> Play $x_i \in \sigma$ and get feedback $M_{x_i} \cdot \theta(t_{j,i}, b) \in \mathbb{R}^{m_{x_i}}$. \\
8: \>\>\> End For\\
9: \>\> End For\\
10: \> {\bf Estimation}\\
11: \> $\tilde{\theta}_{j,i}= M_{\sigma}^{+} (M_{x_1} \cdot \theta(t_{j,1},i), \ldots, M_{x_{|\sigma|}} \cdot \theta(t_{j,|\sigma|}, i)) \in \mathbb{R}^n$.\\
12: \> $\hat{\theta}(b)=\dfrac{\sum_{i=1}^b \sum_{j=1}^{i^{\beta}} \tilde{\theta}_{j,i} }{ \sum_{j=1}^b j^{\beta}} \in \mathbb{R}^n$.\\
13: \> $x(b)\in \argmax_{x \in \mathcal{X}} \bar{r}(x, \hat{\theta}(b))$.\\
14: \> {\bf Exploitation}\\
15: \>\> For $i = 1$ to $\exp(C(b^{\alpha}))$\\
16:\>\>\> Play $x(b)$.\\
17:\>\> End For\\
18: End For
\end{tabbing}
\end{algorithm} 


It is easy to see that the estimators in Algorithm~\ref{alg:PhaseForced} have the following properties: $\E_p [\tilde{\theta}_{j,i}]= M_{\sigma}^{+} (M_{x_1} \cdot \theta^*_p, \ldots, M_{x_{|\sigma|}} \cdot \theta^*_p)= M_{\sigma}^+ M_{\sigma} \cdot \theta^*_p = \theta^*_p$ and hence $\E_p [\hat{\theta}]= \theta^*_p$. Using the fact that $M_{\sigma}^+= (M_{\sigma}^{\top} M_{\sigma})^{-1} M_{\sigma}^{\top}$, we also have the following bound on estimation error of $\theta^*_p$:
\begin{equation}
\begin{aligned}
\label{eq:global-bound}
& \|\tilde{\theta}_{j,i} - \theta^*_p\|_2  \le \|M_{\sigma}^{+} (M_{x_1} \cdot \theta(t_{j,1},i), \ldots, M_{x_{|\sigma|}} \cdot \theta(t_{j,|\sigma|}, i)) -  M_{\sigma}^+ M_{\sigma} \theta^*_p \|_2\\
& = \|(M_{\sigma}^{\top} M_{\sigma})^{-1} \sum_{k=1}^{|\sigma|} M_{x_k}^{\top} M_{x_k} \cdot (\theta(t_{j,k},i)- \theta^*_p)\|_2 
 \le \sqrt{n} \sum_{k=1}^{|\sigma|} \|(M_{\sigma}^{\top} M_{\sigma})^{-1} M_{x_k}^{\top} M_{x_k} \|_2 =: \beta_{\sigma}
\end{aligned}
\end{equation}
where the constant $\beta_{\sigma}$ defined above  depends only on the structure of the linear transformation matrices of the global observer set and not on adversary strategy $p$.

Our first result is about the regret of Algorithm~\ref{alg:PhaseForced} when within phase number $b$, the exploration part spends $|\sigma|$ rounds (constant w.r.t. $b$) and the exploitation part grows polynomially with $b$.

\begin{theorem}
\label{thm:distributionindependentregret}
{\bf (Distribution Independent Regret)} When Algorithm~\ref{alg:PhaseForced} is initialized with the parameters $C(a)= \log a$, $\alpha=1/2$ and $\beta=0$, and the online game is played over $T$ rounds, we get the following bound on expected regret:
\begin{equation}
R(T) \le R_{max} |\sigma| T^{2/3} + 2 R \beta_{\sigma} T^{2/3} \sqrt{ \log 2e^2  + 2 \log T}  + R_{max}
\end{equation}
where $\beta_{\sigma}$ is the constant as defined in Eq.~\ref{eq:global-bound}.
\end{theorem}

Our next result is about the regret of Algorithm~\ref{alg:PhaseForced} when within phase number $b$, the exploration part spends $|\sigma| \cdot b $ rounds (linearly increasing with $b$) and the exploitation part grows exponentially with $b$.

\begin{theorem}
\label{thm:distributiondependentregret1}
{\bf (Distribution Dependent Regret)} When Algorithm~\ref{alg:PhaseForced} is initialized with the parameters $C(a)= h \cdot a$, for a tuning parameter $h>0$, $\alpha=1$ and $\beta=1$, and the online game is played over $T$ rounds, we get the following bound on expected regret:
\begin{equation}
R(T) \le \sum_{x \in \sigma} \Delta_x \left(\dfrac{\log T}{h} \right)^2 + \dfrac{4 \sqrt{2 \pi} e^2 R \Delta_{max} \beta_{\sigma}}{\Delta} e^{\frac{h^2 (2 R^2 \beta^2_{\sigma})}{\Delta^2}} .
\end{equation}
\end{theorem}

Such an explicit bound for a PEGE algorithm that is polylogarithmic in $T$ and explicitly states the multiplicative and additive constants involved in not known, to the best of our knowledge, even in the bandit literature (e.g., earlier bounds \cite{agrawal1989certainty} are asymptotic) whereas here we prove it in the CPM setting. Note that the additive constant above, though finite, blows up exponentially fast as $\Delta \to 0$ for a fixed $h$. It is well behaved however, if the tuning parameter $h$ is on the same scale as $\Delta$. This line of thought motivates us to estimate the gap to within constant factors and then feed that estimate into a PEGE algorithm. This is what we will do in the next section.

\section{Combining Gap Estimation with PEGE}
\label{gap}

\floatstyle{ruled}
\newfloat{algorithm}{htbp}{loa}
\floatname{algorithm}{Algorithm}
\begin{algorithm}
\caption{Algorithm for Gap Estimation}
\label{alg:gap}
\begin{tabbing}
tabs \= tabs \= tabs \= tabs \= tabs \kill
1: Inputs: $T_0$ (exploration threshold) and $\delta$ (confidence parameter) \\\\
2: For  \=$b$ = 1, 2, \ldots,\\
3: \> {\bf Exploration}\\
4: \> \> For $i$ = 1 to $|\sigma|$ \\
5: \>\>\> (Denote) $t_{i} = t$ and $\theta(t_{i},b)= \theta(t)$ ($t$ is current time point).\\
6: \>\>\> Play $x_i \in \sigma$ and get feedback $M_{x_i} \cdot \theta(t_{i}, b) \in \mathbb{R}^{m_{x_i}}$. \\
7: \>\> End For\\
8: \> {\bf Estimation}\\
9: \> $\tilde{\theta}_{b}= M_{\sigma}^{+} (M_{x_1} \cdot \theta(t_{1},b), \ldots, M_{x_{|\sigma|}} \cdot \theta(t_{|\sigma|}, b)) \in \mathbb{R}^n$.\\
10: \> $\hat{\theta}(b)=\dfrac{\sum_{i=1}^b \tilde{\theta}_{i} }{b} \in \mathbb{R}^n$.\\\\
11: \> {\bf Stopping Rule} ($w(b)$ is defined as in Eq.~\eqref{eq:confboundwidth}) \\
12: \> If $\argmax_{x \in \mathcal{X}} \bar{r}(x, \hat{\theta}(b))$ is unique:\\
13: \>\> $\hat{x}(b) =  \argmax_{x \in \mathcal{X}} \bar{r}(x, \hat{\theta}(b))$ \\
14: \>\> $\hat{x}_{-}(b) = \argmax_{x \in \mathcal{X}, x \neq \hat{x}(b)} \bar{r}(x, \hat{\theta}(b))$ (need not be unique) \\
15: \>\> If $\bar{r}(\hat{x}(b), \hat{\theta}(b)) - \bar{r}(\hat{x}_{-}(b), \hat{\theta}(b)) >  6w(b)$: \\
16: \>\>\> STOP and output $\hat{\Delta} = \bar{r}(\hat{x}(b), \hat{\theta}(b)) - \bar{r}(\hat{x}_{-}(b), \hat{\theta}(b))$ \\
17: \>\> End If \\
18: \> End If \\
19: \> If $b > T_0$:\\
20: \> \> STOP and output ``threshold exceeded'' \\
21: \> End If \\
22: End For 
\end{tabbing}
\end{algorithm}

Algorithm~\ref{alg:gap} tries to estimate the gap $\Delta$ to within a constant multiplicative factor. However, if there is no unique optimal action or when the true gap is small, gap estimation
can take a very large amount of time. To prevent that from happening, the algorithm also takes in a threshold $T_0$ as input and definitely stops if the threshold is reached.
The result below assures us that, with high probability, the algorithm behaves as expected. That is, if there is a unique optimal action and the gap is large enough to be estimated
with a given confidence before the threshold $T_0$ kicks in, it will output an estimate $\hat{\Delta}$ in the range $[0.5\Delta, 1.5\Delta]$. On the other hand, if there is no unique optimal action, it does not generate an estimate of $\Delta$ and instead runs out of the exploration budget $T_0$.

\begin{theorem}
\label{thm:gap}
{\bf (Gap Estimation within Constant Factors)} Let $T_0 \ge 1$ and $\delta \in (0,1)$ and define
$
T_1(\delta) = \frac{256 R^2 \beta_\sigma^2}{\Delta^2} \log \frac{512 e^2 R^2 \beta_\sigma^2}{\Delta^2 \delta}
$,
$
T_2(\delta) = \frac{16 R^2 \beta_\sigma^2}{\Delta^2} \log \frac{4e^2}{\delta} 
$.
Consider Algorithm~\ref{alg:gap} run with
\begin{equation}\label{eq:confboundwidth}
w(b) = \sqrt{\dfrac{R^2 \beta_{\sigma}^2 \log (\frac{4e^2b^2}{\delta})}{b}} .
\end{equation}
Then, the following 3 claims hold.
\begin{enumerate}
\item
Suppose Assumption 4 holds and $T_1(\delta) < T_0$. Then with probability at least $1-\delta$, Algorithm~\ref{alg:gap} stops in $T_1(\delta)$ episodes and outputs an estimate $\hat{\Delta}$ that satisfies
$\frac{1}{2} \Delta \le \hat{\Delta} \le \frac{3}{2} \Delta$.
\item
Suppose Assumption 4 holds and $T_0 \le T_1(\delta)$. Then with probability at least $1-\delta$, the algorithm either outputs ``threshold exceeded'' or outputs an estimate $\hat{\Delta}$ that satisfies
$\frac{1}{2} \Delta \le \hat{\Delta} \le \frac{3}{2} \Delta$. Furthermore, if it outputs $\hat{\Delta}$, it must be the case that the algorithm stopped at an episode $b$ such that $T_2(\delta) < b < T_0$.
\item
Suppose Assumption 4 fails. Then, with probability at least $1-\delta$, Algorithm~\ref{alg:gap} stops in $T_0$ episodes and outputs ``threshold exceeded''.
\end{enumerate}
\end{theorem}

Equipped with Theorem~\ref{thm:gap}, we are now ready to combine Algorithm~\ref{alg:gap} with Algorithm~\ref{alg:PhaseForced} to give Algorithm~\ref{alg:combined}. Algorithm~\ref{alg:combined} first calls Algorithm~\ref{alg:gap}.
If Algorithm~\ref{alg:gap} outputs an estimate $\hat{\Delta}$ it is fed into Algorithm~\ref{alg:PhaseForced}. If the threshold $T_0$ is exceeded, then the remaining time is spent in pure exploitation.
Note that by choosing $T_0$ to be of order $T^{2/3}$ we can guarantee a worst case regret of the same order even when unique optimality assumption fails. For PM games that are globally
observable but not locally observable, such a distribution independent $O(T^{2/3})$ bound is known to be optimal \cite{bartok2013}.

\begin{theorem}
\label{thm:combined}
{\bf (Regret Bound for PEGE2)}
Consider Algorithm~\ref{alg:combined} run with knowledge of the number $T$ of rounds.
Consider the distribution independent bound
\[
B_1(T) = 2(2 R \beta_\sigma |\sigma|^2 R^2_{max} T)^{2/3} \sqrt{\log (4e^2 T^3}) + R_{max} ,
\]
and the distribution dependent bound
\[
B_2(T) = \frac{256 R^2 \beta_\sigma^2}{\Delta^2} \log \frac{512 e^2 R^2 \beta_\sigma^2 T}{\Delta^2} R_{max} |\sigma|
+ \sum_{x \in \sigma} \Delta_x \frac{36 R^2 \beta_\sigma^2 \log T}{\Delta^2} + \frac{8 e^2 R^2 \beta_\sigma^2}{\Delta^2} + R_{max} .
\]
If Assumption 4 fails, then the expected regret of Algorithm~\ref{alg:combined} is bounded as
$
R(T) \le B_1(T)
$.
If Assumption 4 holds, then the expected regret of Algorithm~\ref{alg:combined} is bounded as
\begin{equation}\label{eq:cases}
R(T) \le \begin{cases}
B_2(T) & \text{ if } \ T_1(\delta) < T_0 \\
O(T^{2/3} \log T) & \text{ if } \ T_0 \le T_1(\delta)
\end{cases} ,
\end{equation}
where $T_1(\delta)$ is as defined in Theorem \ref{thm:gap} and $\delta, T_0$ are as defined in Algorithm \ref{alg:combined}.
\end{theorem}

In the above theorem, note that $T_1(\delta)$ scales as $\Theta(\tfrac{1}{\Delta^2} \log \tfrac{T}{\Delta^2})$ and $T_0$ as
$\Theta(T^{2/3})$. Thus, the two cases in Eq.~\eqref{eq:cases} correspond to large gap and small gap situations respectively.

\floatstyle{ruled}
\newfloat{algorithm}{htbp}{loa}
\floatname{algorithm}{Algorithm}
\begin{algorithm}
\caption{Algorithm Combining PEGE with Gap Estimation (PEGE2)}
\label{alg:combined}
\begin{tabbing}
tabs \= tabs \= tabs \= tabs \= tabs \kill
1: Input: $T$ (total number of rounds) \\
\\
2: Call Algorithm \ref{alg:gap} with inputs $T_0 = \left(\frac{2R\beta_\sigma T}{|\sigma| R_{max}}\right)^{2/3}$ and $\delta = 1/T$ \\
3: If Algorithm \ref{alg:gap} returns ``threshold exceeded'': \\
4: \> Let $\hat{\theta}(T_0)$ be the latest estimate of $\theta^*_p$ maintained by Algorithm~\ref{alg:gap} \\
5: \> Play $\hat{x}(T_0) = \argmax_{x \in \mathcal{X}} \bar{r}(x, \hat{\theta})$ for the remaining $T-T_0|\sigma|$ rounds \\
6: Else: \\
7: \> Let $\hat{\Delta}$ be the gap estimate produced by Algorithm~\ref{alg:gap} \\
8: \> For all remaining time steps, run Algorithm \ref{alg:PhaseForced} with parameters $C(a)= h a$ with \\
\> $h = \frac{\hat{\Delta}^2}{9 R^2 \beta^2_{\sigma}}$, $\alpha=1$, $\beta=0$\\
9: End If
\end{tabbing}
\end{algorithm}

\section{Comparison with GCB Algorithm}
\label{comparison}

We provide a detailed comparison of our results with those obtained for GCB \cite{lincombinatorial2014}.
\begin{enumerate*}[label=(\alph*)]
\item
While we use the same CPM model, our work is inspired by the forced exploration technique while GCB is inspired by the confidence bound technique, both of which are classic in the bandit literature.
\item
One instantiation of our PEGE framework gives an $O(T^{2/3}\sqrt{\log T})$ distribution independent regret bound (Theorem~\ref{thm:distributionindependentregret}), which does not require call to arg-secondmax oracle. This is of substantial practical advantage over GCB since even for linear optimization problems over polyhedra, standard routines usually do not have option of computing action(s) that achieve second maximum value for the objective function. 
\item
Another instantiation of the PEGE framework gives an $O(\log^2 T)$ distribution dependent regret bound (Theorem~\ref{thm:distributiondependentregret1}), which neither requires call to arg-secondmax oracle nor the assumption of existence of unique optimal action for learner. This is once again important, since the assumption of existence of unique optimal action might be impractical, especially for exponentially large action space. However, the caveat is that improper setting of the tuning parameter $h$ in Theorem~\ref{thm:distributiondependentregret1} can lead to a large additive component in the regret.
\item
A crucial point, which we had highlighted in the beginning, is that the regret bounds achieved by PEGE and PEGE2 do not have dependence on size of learner's action space, i.e., $|\mathcal{X}|$. The dependence is only on the size of global observable set $\sigma$, which is guaranteed to be not more than dimension of adversary's action space. Thus, though we have adopted the CPM model \cite{lincombinatorial2014}, our algorithms achieve meaningful regret bounds for countably infinite or even continuous learner's action space. In contrast, the GCB regret bounds have explicit, logarithmic dependence on size of learner's action space. Thus, their results cannot be extended to problems with infinite learner's action space (see Section~\ref{application} for an example), and are restricted to large, but \emph{finite} action spaces.
\item
The PEGE2 algorithm is a true analogue of the GCB algorithm, matching the regret bounds of GCB in terms of $T$ and gap $\Delta$ with the advantage that it has no dependence on $|\mathcal{X}|$. The disadvantage, however, is that PEGE2 requires knowledge of time horizon $T$, while GCB is an anytime algorithm.
\end{enumerate*}

\section{Application to Online Ranking}
\label{application}
A recent paper studied an interesting problem of online ranking with feedback restricted to top ranked items \cite{chaudhuri2015}. The problem was studied in a non-stochastic setting, i.e., it was assumed that an oblivious adversary generates reward vectors. Moreover, the learner's action space was exponentially large in number of items to be ranked. The paper made the connection of the problem setting to PM games (but not combinatorial PM games) and proposed an efficient algorithm for the specific problem at hand. However, a careful reading of the paper shows that their algorithmic techniques can handle the CPM model we have discussed so far, but in the \emph{non-stochastic} setting. The reward function is linear in both learner's and adversary's moves, adversary's move is restricted to a finite space of vectors and feedback is a linear transformation of adversary's move. In this section, we give a brief description of the problem setting and show how our algorithms can be used to efficiently solve the problem of online ranking with feedback on top ranked items in the \emph{stochastic} setting. We also give an example of how the ranking problem setting can be somewhat naturally extended to one which has continuous action space for learner, instead of large but finite action space.

The paper considered an online ranking problem, where a learner repeatedly re-ranks a set of $n$, fixed items, to satisfy diverse users' preferences, who visit the system sequentially. Each learner action $x$ is a permutation of the $n$ items. Each user has like/dislike preference for each item, varying between users, with each user's preferences encoded as $n$ length binary relevance vectors $\theta$. Once the ranked list of items is presented to the user, the user scans through the items, but gives relevance feedback only on top ranked item. However, the performance of the learner is judged based on full ranked list and unrevealed, full relevance vector. Thus, we have a PM game, where neither adversary generated relevance vector nor reward is revealed to learner. The paper showed how a number of practical ranking measures, like Discounted Cumulative Gain (DCG), can be expressed as a linear function, i.e., $r(x, \theta)= f(x) \cdot \theta$. The practical motivation of the work was based on learning a ranking strategy to satisfy diverse user preferences, but with limited feedback received due to user burden constraints and privacy concerns.

{\bf Online Ranking with Feedback at Top as a Stochastic CPM Game.} We show how our algorithms can be applied in online ranking with feedback for top ranked items by showing how it is a specific instance of the CPM model and how our key assumptions are satisfied. The learner's action space is the finite but exponentially large space of  $\mathcal{X}=n!$ permutations. Adversary's move is an $n$ dimensional relevance vector, and thus, is restricted to $\{0,1\}^n$ (finite space of size $2^n$) contained in $[0,1]^n$. In the stochastic setting, we can assume that adversary samples $\theta \in \{0,1\}^n$  from a fixed distribution on the space. Since the feedback on playing a permutation is the relevance of top ranked item, each move $x$ has an associated transformation matrix $M_x \in \{0,1\}^{n}$, with $1$ in the place of the item which is ranked at the top by $x$ and $0$ everywhere else. Thus, $M_x \cdot \theta$ gives the relevance of item ranked at the top by $x$. The global observable set $\sigma$ is the set of any $n$ actions, where each action, in turn, puts a distinct item on top. Hence, $M_{\sigma}$ is the $n \times n$ dimensional permutation matrix.  Assumption 1 is satisfied because the reward function is linear in $\theta$ and $\bar{r} (x, \theta^*_p)= f(x) \cdot \theta^*_p$, where $\E_p [\theta]= \theta^*_p \in [0,1]^n$. Assumption 2 is satisfied since there will always be a global observable set of size $n$ and can be found easily. In fact, there will be multiple global observable sets, with the freedom to choose any one of them. Assumption 3 is satisfied due to the expected reward function being linear in second argument. The Lipschitz constant is $\max_{x \in \mathcal{X}}\|f(x)\|_2$, which is always less than some small polynomial factor of $n$, depending on specific $f(\cdot)$. The value of $\beta_{\sigma}$ can be easily seen to be $n^{3/2}$. The $\argmax$ oracle returns the permutation which simply sorts items according to their corresponding $\theta$ values. The arg-secondmax oracle is more complicated, though feasible. It requires first sorting the items according to $\theta$ and then compare each pair of consecutive items to see where least drop in reward value occurs. 

{\bf Likely Failure of Unique Optimal Action Assumption.} Assumption 4 is unlikely to hold in this problem setting. The mean relevance vector $\theta^*_p$ effectively reflects the average preference of all users for each of the $n$ items. It is very likely that at least a few items will not be liked by anyone and which will ultimately be always ranked at the bottom. Equally possible is that two items will have same user preference on average, and can be exchanged without hurting the optimal ranking. Thus, existence of an unique optimal ranking, which indicates that each item will have different average user preference than every other item, is unlikely. Thus, PEGE algorithm can still be applied to get poly-logarithmic regret (Theorem~\ref{thm:distributiondependentregret1}), but GCB will only achieve $O(T^{2/3}\log T)$ regret. 

{\bf A PM Game with Infinite Learner Action Space.} We give a simple modification of the ranking problem above to show how the learner can have continuous action space. The learner now ranks the items by producing an $n$ dimensional score vector $x \in [0,1]^n$ and sorting items according to their scores. Thus the learner's action space is now an uncountably infinite continuous space. As before, the user gets to see the ranked list and gives relevance feedback on top ranked item. The learner's performance will now be judged by a continuous loss function, instead of a discrete-valued ranking measure, since its moves are in a continuous space. Consider the simplest loss, viz., the squared ``loss'' $r(x, \theta)= -\|x- \theta\|^2_2$ (note -ve sign to keep reward interpetation). It can be easily seen that $\bar{r}(x, \theta^*_p) = E_{\theta \sim p} [r(x,\theta)] = -\|x\|_2^2+2 x \cdot \theta^*_p-\mathbf{1} \cdot \theta^*_p$,  if the relevance vectors $\theta$ are in $\{0,1\}^n$. Thus, the Lipschitz condition is satisfied. The global observable set is still of size $n$, with the $n$ actions being any $n$ score vectors, whose sorted orders place each of the $n$ items, in turn, on top. $\beta_{\sigma}$ remains same as before, with $\argmax_{x} \E_{\theta\sim p}r(x, \theta)= \E_{\theta\sim p}[\theta] = \theta^*_p$. 


\bibliographystyle{unsrt}
\bibliography{PM}

\newpage

\section{Appendix}
\label{appendix}

We first state the large deviation inequality for vector-valued martingales, which is the generalization of Azuma-Hoeffding inequality for scalar valued martingales.

{\bf Theorem 1.8 of \cite{hayes2005large}}: Let $X_0, X_1, \ldots, X_m$ be a weak martingale sequence taking values in euclidean space $\mathbb{R}^d$, with $\E[X_i | X_{i-1}] = X_{i-1}$. Let $X_0=0$ and $\|X_i -X_{i-1}\|_2 \le 1$, for $i=1,2,\ldots,m$. Then, for every $\epsilon>0$,
\begin{equation}
\Pr [\|X_m\|_2 \ge \epsilon] < 2 e^2 e^{\frac{- \epsilon^2}{ m}}
\end{equation}

We use the concentration inequality to get a uniform confidence bound, over the space of learner's action, on the deviation of estimated reward from true reward, after each estimate of mean reward vector is produced.

\begin{lemma}
\label{lem:conc-bound}
At the end of exploration phase within phase $b$, $b=1,2,\ldots$, of Algorithm PEGE,  the estimator of reward vector $\theta^*_p$ is $\hat{\theta}(b)=  \dfrac{\sum_{i=1}^b \sum_{j=1}^{i^{\beta}} \tilde{\theta}_{j,i} }{ \sum_{j=1}^b j^{\beta}}$. Then, $\forall \ \eta>0$, 
\begin{equation}
\label{eq:conc-bound}
 \Pr[\forall \ x \in \mathcal{X}: |\bar{r}(x, \hat{\theta}(b))- \bar{r}(x, \theta^*_p)| \le \eta] \ge 1- 2 e^2 e^{\frac{-(\sum_{i=1}^{b^{\beta}} i^{\beta}) \eta^2}{R^2 \beta_{\sigma}^2}}
\end{equation}

where $\beta_{\sigma}$ is the constant as defined in Eq.~\ref{eq:global-bound} and $R$ is the Lipschitz constant defined in Assumption 3.
\end{lemma}
 
 \begin{proof}
 Let $\{X_{i,j}\}_{\underset{i=1,\ldots, b}{j=1,\ldots, i^{\beta}}}$ be a sequence of random vectors, defined as follows:
 \begin{equation}
 X_{i,j}= \dfrac{\sum_{i'=1}^ {i-1} \sum_{j'=1}^{{i'}^{\beta}} \theta^*_p + \sum_{k=1}^j \theta^*_p - (\sum_{i'=1}^ {i-1} \sum_{j'=1}^{{i'}^{\beta}} \tilde{\theta}_{j',i'} + \sum_{k=1}^j \tilde{\theta}_{k,i})}{\sum_{i''=1}^ b \sum_{j''=1}^{(i'')^{\beta}} \beta_{\sigma}}
 \end{equation}
It can be checked that the $\ell_2$ norm of the difference between any two consecutive random vectors is bounded by a constant. That is, $\|X_{i,j}- X_{i, j-1}\|_2 = \dfrac{\|\theta^*_p - \tilde{\theta}_{j,i}\|_2}{\sum_{i''=1}^ b \sum_{j''=1}^{(i'')^{\beta}} \beta_{\sigma}} \le \dfrac{1}{\sum_{i''=1}^ b (i'')^{\beta}}$ and $\|X_{i+1, 1} - X_{i, i^{\beta}}\|_2 \dfrac{\|\theta^*_p - \tilde{\theta}_{1,i+1}\|_2}{\sum_{i''=1}^ b \sum_{j''=1}^{(i'')^{\beta}} \beta_{\sigma}} \le \dfrac{1}{\sum_{i''=1}^ b (i'')^{\beta}}$. 

Also, $\tilde{\theta}_{j,i}$ is independent of all estimators formed before $\tilde{\theta}_{j,i}$ in Algorithm PEGE. Thus,
\begin{equation}
\begin{aligned}
\E[X_{i,j} - X_{i,j-1}| X_{i,j-1}]  & = \E \left[\frac{\theta^*_p - \tilde{\theta}_{j,i}}{\sum_{i''=1}^ b \sum_{j''=1}^{(i'')^{\beta}} \beta_{\sigma}}|X_{i,j-1} \right]\\
& = \E \left[\frac{\theta^*_p - \tilde{\theta}_{j,i}}{\sum_{i''=1}^ b \sum_{j''=1}^{(i'')^{\beta}} \beta_{\sigma}} \right] = 0
\end{aligned}
\end{equation}

Thus, $\{X_{i,j}\}_{\underset{i=1,\ldots, b}{j=1,\ldots, i^{\beta}}}$ satisfy the criteria of weak martingale sequence and hence, by the large deviation inequality of vector valued martingales, we have, $\forall \ \epsilon>0$, $\Pr[\|X_{b,b^{\beta}}\|_2 \ge \epsilon]< 2 e^2 e^{\left(\frac{- \epsilon^2}{ \frac{\sum_{i=1}^b \sum_{j=1}^{i^{\beta}}1}{(\sum_{i=1}^b i^{\beta})^2}}\right)}= 2 e^2 e^{-\epsilon^2 \sum_{i=1}^b i^{\beta}}$.\\

Now, it can be clearly seen that $\|\theta^*_p -\hat{\theta}(b)\|_2 = \beta_{\sigma} \|X_{b,b^{\beta}}\|_2$ and let $\eta= \beta_{\sigma} \epsilon$. Then, $\forall \ \eta>0$, we get $\Pr[ \|\theta^*_p - \hat{\theta}(b)\|_2 \ge \eta] \le 2 e^2 e^{\frac{-(\sum_{i=1}^{b^{\beta}} i^{\beta}) \eta^2}{\beta_{\sigma}^2}}$.

Using the Lipschitz property of expected reward function (Assumption 3), we have 

\begin{equation}
\begin{aligned}
\Pr(\exists\ x \in \mathcal{X}: |\bar{r}(x, \hat{\theta}(b))- \bar{r}(x, \theta^*_p)| \ge \eta) & \le \Pr ( R \cdot  \|\theta^*_p - \hat{\theta}(b)\|_2 \ge \eta)\\
& \le 2 e^2 e^{\frac{-(\sum_{i=1}^{b^{\beta}} i^{\beta}) \eta^2}{R^2 \beta_{\sigma}^2}}
\end{aligned}
\end{equation}

Taking complement of the event completes the proof.
 \end{proof}
 
\subsection{Proof of Results in Section \ref{algorithmanalysis}}

\subsubsection{Proof of Theorem~\ref{thm:distributionindependentregret}}

We first restate the theorem. 

{\bf Distribution Independent Regret:} When Algorithm PEGE is initialized with the parameters $C(a)= \log a$, $\alpha=1/2$ and $\beta=0$, and the online game is played over $T$ rounds, we get the following bound on expected regret:
\begin{equation}
R(T) \le R_{max} |\sigma| T^{2/3} + 2 R \beta_{\sigma} T^{2/3} \sqrt{ \log 2e^2  + 2 \log T}  + R_{max}
\end{equation}
where $\beta_{\sigma}$ is the constant as defined in Eq.~\ref{eq:global-bound}.

\begin{proof}

Let Algorithm PEGE run for $K$ phases, with parameters initialized as $C(a)= \log a$, $\alpha=1/2$ and $\beta=0$.

{\bf Exploration regret}: During every exploration phase, the expected regret is bounded by $|\sigma| R_{max}$, where $R_{max}$ is as given in Assumption 3. Thus, total expected regret due to exploration is $K |\sigma| R_{max}$.\\

{\bf Exploitation regret}: Let $x^* \in \argmax_{x \in \mathcal{X}} \bar{r}(x, \theta^*_p)$ and $x(b) \in \argmax_{x \in \mathcal{X}} \bar{r}(x, \hat{\theta}(b))$. During every exploitation round within phase $b$ of Algorithm PEGE, the expected regret is $|\bar{r}(x(b), \theta^*_p) - \bar{r}(x^*, \theta^*_p)|$ .

Now, from Lemma~\ref{lem:conc-bound}, with $\beta=0$, the following holds w.p. $\ge 1 - \delta_b$, 
\begin{equation}
\label{eq:deviation}
\forall \ x, \ |\bar{r}(x,\theta^*_p) - \bar{r}(x, \hat{\theta}(b)| \le \underbrace{\sqrt{\dfrac{R^2 \beta_{\sigma}^2 \log (\frac{2e^2}{\delta_b})}{b}}}_{\eta_b}
\end{equation}

Then, w.p. $\ge 1-  \delta_b$, the following event holds true: $|\bar{r}(x^*, \theta^*_p) - \bar{r}(x(b), \theta^*_p)| \le 2 \eta_b$, as explained:
\begin{equation*}
\begin{aligned}
\bar{r}(x^*, \theta^*_p) & \le \bar{r}(x^*, \hat{\theta}(b)) + \eta_b \ \ \text{from Eq.}~\ref{eq:deviation} \\
& \le \bar{r}(x(b), \hat{\theta}(b)) + \eta_b\\
& \le \bar{r}(x(b), \theta^*_p) + 2 \eta_b \ \ \text{from Eq.}~\ref{eq:deviation}
\end{aligned}
\end{equation*}

Thus, the event $|\bar{r}(x^*, \theta^*_p) - \bar{r}(x(b), \theta^*_p)| \le 2 \eta_b$ holds true w.p. $\ge 1-  \delta_b$, for every fixed phase $b$. Then, w.p. $\ge 1- \sum_{i=1}^K \delta_i$, the following holds true:
\begin{equation*}
\forall \ b,  |\bar{r}(x^*, \theta^*_p) - \bar{r}(x(b), \theta^*_p)| \le 2 \eta_b
\end{equation*}

Note that the expected regret per round is always bounded by $R_{\max}$ (since expected reward is bounded by $R_{max}$) .


The number of rounds of exploitation in phase $b$ is $b^{\alpha}$. Hence, the total expected regret due to exploitation, over $K$ phases is:

\begin{equation*}
\sum_{i=1}^K \left(\underbrace{ (1-  \sum_{j=1}^K \delta_j) \frac{2 R \beta_{\sigma} \sqrt{\log(2e^2/\delta_i)}}{\sqrt{i}} + (\sum_{j=1}^K \delta_j) R_{\max}}_{\text{expected regret per exploitation round}}  \right) i^{\alpha}
\end{equation*}

Taking $\delta_1=\delta_2= \ldots= \delta_K= \delta$,  and summing over exploration and exploitation regret over $K$ phases, we get

\begin{equation}
R(T) \le K |\sigma| R_{max} + \sum_{i=1}^K \left( (1- K \delta) \frac{2 R \beta_{\sigma} \sqrt{\log(2e^2/\delta)}}{\sqrt{i}} + (K \delta) R_{\max}  \right) i^{\alpha}
\end{equation}\\

Using the inequality $\sum_{i=1}^K i^y \le \int_{0}^{K} i^y dy \le K^{y+1}$, we get expected regret:
\begin{equation}
\label{eq:regret}
R(T) \le K |\sigma| R_{max} +   (1- K\delta) 2 R \beta_{\sigma} \sqrt{\log(2e^2/\delta)} K^{\alpha + 1/2} +  K\delta R_{\max} K^{\alpha + 1}
\end{equation}

Now, we relate $K$ to total time $T$ as:  $T=  |\sigma| K + \sum_{i=1}^K i^{\alpha} \sim K^{\alpha+1}$, for large $K$. 

Hence $ K \sim T^{\frac{1}{1+ \alpha}}$. Substituting value of $K$ in Eq~\ref{eq:regret}, and taking $\alpha={1/2}$ and $\delta= \frac{1}{KT}$ gives us the required bound on expected regret. 
\end{proof}

Our next lemma shows that as the number of phases $b$ grows in Algorithm PEGE, the probability of selecting a sub-optimal arm for greedy exploitation shrinks.
 
 \begin{lemma}
 \label{lem:suboptimal-arm}
 At the end of exploration phase within phase $b$, $b=1,2,\ldots$, the estimator constructed is $\hat{\theta}(b)=  \dfrac{\sum_{i=1}^b \sum_{j=1}^{i^{\beta}} \tilde{\theta}_{j,i} }{ \sum_{j=1}^b j^{\beta}}$. Then the following holds,
 \begin{equation}
 \Pr ( \argmax_{x \in \mathcal{X}} \bar{r}(x, \hat{\theta}(b)) \not\subseteq \argmax_{x \in \mathcal{X}} \bar{r}(x, \theta^*_p)) \le 2 e^2 e^{\frac{-(\sum_{i=1}^{b^{\beta}} i^{\beta}) \Delta^2}{4 R^2 \beta_{\sigma}^2}}
 \end{equation}
 \end{lemma}
 
 \begin{proof}
 Let us assume $x' \in \argmax_{x \in \mathcal{X}} \bar{r}(x, \hat{\theta}(b))$ such that $x' \notin \argmax_{x \in \mathcal{X}} \bar{r}(x, \theta^*_p)$. Let $x^* \in \argmax_{x \in \mathcal{X}} \bar{r}(x, \theta^*_p)$. Then, by our assumption, $\bar{r}(x', \hat{\theta}(b)) \ge \bar{r}(x^*, \hat{\theta}(b))$. By definition of gap $\Delta$, we also have $\bar{r}(x^*, \theta^*_p) - \bar{r}(x', \theta^*_p) \ge \Delta$. The two inequalities imply that at least one of the following two inequalities has to hold: either $|\bar{r}(x^*, \theta^*_p)- \bar{r}(x^*, \hat{\theta}(b))| \ge \frac{\Delta}{2}$ or $|\bar{r}(x', \hat{\theta}(b))- \bar{r}(x', \theta^*_p)| \ge \frac{\Delta}{2}$.
 
 Thus, $\argmax_{x \in \mathcal{X}} \bar{r}(x, \hat{\theta}(b)) \not\subseteq \argmax_{x \in \mathcal{X}} \bar{r}(x, \theta^*_p) \implies \exists \ x \in \mathcal{X}: |\bar{r}(x, \theta^*_p)- \bar{r}(x, \hat{\theta}(b))| \ge \frac{\Delta}{2}$. By using Lemma~\ref{lem:conc-bound}, and substituting $\eta= \frac{\Delta}{2}$, we get our result.
 \end{proof}

\subsubsection{Proof of Theorem~\ref{thm:distributiondependentregret1}}

We restate the theorem before proving:

{\bf Distribution Dependent Regret:} When Algorithm PEGE is initialized with the parameters $C(a)= h \cdot a$, for a tuning parameter $h>0$, $\alpha=1$ and $\beta=1$, and the online game is played over $T$ rounds, we get the following bound on expected regret:
\begin{equation}
R(T) \le \sum_{x \in \sigma} \Delta_x \left(\dfrac{\log T}{h} \right)^2 + \dfrac{4 \sqrt{2 \pi} e^2 R \Delta_{max} \beta_{\sigma}}{\Delta} e^{\frac{h^2 (2 R^2 \beta^2_{\sigma})}{\Delta^2}} .
\end{equation}
 
 \begin{proof}
 
Let total number of phases that the algorithm runs for be $K$. We relate $K$ to total time $T$ as (after substituting parameters $C(a)= h \cdot a$, $\alpha=1$ and $\beta=1$ in Algorithm PEGE):

$T = \sum_{i=1}^K |\sigma| i + \sum_{i=1}^K e^{h  i} \ge e^{h  K}$ $\implies$ $ K \le \dfrac{\log T}{h}$.
 
{\bf Exploration regret}: Sine we are in distribution dependent setting now, expected exploration regret in each exploration phase is $\sum_{x \in \sigma} \Delta_x$.   Hence, total expected exploration regret is upper bounded by:

 $\sum_{i=1}^K (\sum_{ x \in \sigma} \Delta_x) i =  \sum_{ x \in \sigma} \Delta_x \frac{K(K+1)}{2} \le (\sum_{ x \in \sigma} \Delta_x) \dfrac{\log^2 T}{h^2}$.
 
{\bf Exploitation regret}: When a sub-optimal arm is picked in an exploitation round, the expected regret in that round is: $\le \Delta_{max}$. Using Lemma~\ref{lem:suboptimal-arm} with $\beta=1$, the total expected regret due to exploitation over $K$ phases is upper bounded by:
 \begin{equation}
 \begin{aligned}
 \sum_{i=1}^K  \underbrace{  2 e^2 \Delta_{max} \ e^{h i - \frac{i(i+1)}{2} \frac{\Delta^2}{4 R^2 \beta^2_{\sigma}}}}_{\text{expected exploitation regret upper bound in phase i}} &\le 2 e^2 \Delta_{max} \sum_{i=1}^{\infty} e^{h i - \frac{i(i+1)}{2} \frac{\Delta^2}{4 R^2 \beta^2_{\sigma}}}\\
 & \le 2 e^2 \Delta_{max} \int_{- \infty}^{\infty} e^{h y - \frac{y(y+1)}{2} \frac{\Delta^2}{4 R^2 \beta^2_{\sigma}}} dy
 \end{aligned}
 \end{equation}
 The integral is the moment generating function (adjusting for normalization constant) of a gaussian random variable $ Y \in \mathcal{N} (0, \frac{4 R^2 \beta^2_{\sigma}}{\Delta^2})$. Thus, the integral is $\E[e^{hY}]= e^{\frac{2 h^2 R^2 \beta^2_{\sigma}}{\Delta^2 }}$ and total expected regret due to exploitation is upper bounded by: $\dfrac{4 e^2 \Delta_{max} \sqrt{2 \pi} R \beta_{\sigma}}{\Delta} e^{\frac{2 h^2 R^2 \beta^2_{\sigma}}{\Delta^2 }}$.
 
 Summing over exploration and exploitation regrets completes the proof.
 \end{proof}
 
 \subsection{Proof of Results in Section~\ref{gap}}
 
The following theorem is about the version of PEGE that Algorithm~\ref{alg:combined} calls on line 8. It will be needed in the proof of Theorem~\ref{thm:combined}.

\begin{theorem}
\label{thm:distributiondependentregret2}
{\bf (Distribution Dependent Regret, version 2)} When Algorithm~\ref{alg:PhaseForced} is initialized with the parameters $C(a)= h \cdot a$, for a tuning parameter $0<h< \frac{\Delta^2}{4 R^2 \beta^2_{\sigma}}$, $\alpha=1$ and $\beta=0$, and the online game is played over $T$ rounds, we get the following bound on expected regret:
\begin{equation}
R(T) \le \sum_{x \in \sigma} \Delta_x \dfrac{\log T}{h} + \dfrac{2 e^2 \Delta_{max}}{\frac{\Delta^2}{4 R^2 \beta^2_{\sigma}} - h} 
\end{equation}
\end{theorem}

{\bf Note}: Compared to Theorem \ref{thm:distributiondependentregret1}, the regret bound has better dependence on $T$ --- $O(\log T)$ instead of $O(\log^2 T)$ --- but it also has a disadvantage.
If the tuning parameter $h$ is incorrectly set, say $h\ge \frac{\Delta^2}{4 R^2 \beta^2_{\sigma}}$, then the bound does not even apply.

\begin{proof}
{\bf Key Steps}: 
 
 Let total number of phases that the algorithm runs for be $K$.
 First: $T = \sum_{i=1}^K |\sigma|  + \sum_{i=1}^K e^{h  i} \ge e^{h  K}$ $\implies$ $ K \le \dfrac{\log T}{h}$.
 
 Expected regret due to exploration: $\sum_{i=1}^K (\sum_{ x \in \sigma} \Delta_x)  =  \sum_{ x \in \sigma} \Delta_x K \le (\sum_{ x \in \sigma} \Delta_x) \dfrac{\log T}{h}$.
 
 Expected regret due to exploitation: When a sub-optimal arm is picked, expected regret $\le \Delta_{max}$. Using Lemma~\ref{lem:suboptimal-arm} with $\beta=0$, and tuning parameter $h < \frac{\Delta^2}{4 R^2 \beta^2_{\sigma}}$,  we get total expected regret due to exploitation
 \begin{equation}
 \begin{aligned}
 2 e^2 \Delta_{max} \sum_{i=1}^K e^{h i - i \frac{\Delta^2}{4 R^2 \beta^2_{\sigma}}} &\le 2 e^2 \Delta_{max} \sum_{i=1}^{\infty} e^{h i - i \frac{\Delta^2}{4 R^2 \beta^2_{\sigma}}}\\
 & = 2 e^2 \Delta_{max} \sum_{i=1}^{\infty} e^{-i (\frac{\Delta^2}{4 R^2 \beta^2_{\sigma}} - h)}\\
 & \le 2 e^2 \Delta_{max} \int_{0}^{\infty} e^{-y (\frac{\Delta^2}{4 R^2 \beta^2_{\sigma}} - h)}dy\\
 & = \dfrac{2 e^2 \Delta_{max}}{\frac{\Delta^2}{4 R^2 \beta^2_{\sigma}} - h}
 \end{aligned}
 \end{equation}
\end{proof}

\subsubsection{Proof of Theorem~\ref{thm:gap}}

\begin{proof}
Note that Assumption 1 through Assumption 3 hold.
Therefore, from Lemma~\ref{lem:conc-bound}, with $\beta=0$ we get, with probability at least $1 - \delta_b$, 
\begin{equation*}
\forall  x, \ |\bar{r}(x,\theta^*_p) - \bar{r}(x, \hat{\theta}(b)| \le \sqrt{\dfrac{R^2 \beta_{\sigma}^2 \log (\frac{2e^2}{\delta_b})}{b}}
\end{equation*}
Let $\delta_b = \delta/2b^2$ which implies $\sum_{b \ge 1} \delta_b = \pi^2 \delta /12 < \delta$. Thus, setting
$w(b) = \sqrt{\dfrac{R^2 \beta_{\sigma}^2 \log (\frac{4e^2b^2}{\delta})}{b}} ,$
the event $E$ defined as
\begin{equation}\label{eq:boundshold}
\forall b \ge 1, \forall x \in \mathcal{X}, |\bar{r}(x,\hat{\theta}(b)) - \bar{r}(x,\theta^*_p) | \le w(b) .
\end{equation}
holds with probability at least $1-\delta$.

\begin{enumerate}
\item
Note that $b \ge T_1(\delta)$ implies $8w(b) < \Delta$. This is because the latter has the form $e^{L b} > M b$ with $M = 2e/\delta$ and $L = \Delta^2/(128R^2\beta_\sigma^2)$. Setting $b \ge 2/L \log(2M/L)$ guarantees
that $e^{Lb/2} \ge 2M/L$ which implies that $e^{Lb} \ge Mb$ since $e^{Lb/2} \ge 1 + Lb/2 \ge Lb/2$. 

If $8w(b) < \Delta$ then clearly $2w(b) < \Delta$. Let $x \neq x^*$ be arbitrary. We have the following chain of implications:
\begin{align*}
&& 2w(b) &< \Delta \\
&\Rightarrow& 2w(b) &< \bar{r}(x^*, \theta^*_p) - \bar{r}(x^{*}_{-}, \theta^*_p) && \text{(def. of $\Delta$)}\\
&\Rightarrow& 2w(b) &< \bar{r}(x^*, \theta^*_p) - \bar{r}(x, \theta^*_p) && \text{(Assumption 4)}\\
&\Rightarrow& 0 &< \bar{r}(x^*, \hat{\theta}(b)) - \bar{r}(x, \hat{\theta}(b)) . && (\because E \text{ holds}) 
\end{align*}
This means that the If condition on line 12 will evaluate to true and $\hat{x}(b)$ on line 13 will be set to $x^*$.

We also have the following chain of implications:
\begin{align*}
&& 8w(b) &< \Delta \\
&\Rightarrow& 8w(b) &< \bar{r}(x^*, \theta^*_p) - \bar{r}(x^{*}_{-}, \theta^*_p) && \text{(def. of $\Delta$)}\\
&\Rightarrow& 8w(b) &< \bar{r}(x^*, \theta^*_p) - \bar{r}(\hat{x}_{-}(b), \theta^*_p) && (\because \bar{r}(\hat{x}_{-}(b), \theta^*_p) \le \bar{r}(x^{*}_{-}, \theta^*_p)) \\
&\Rightarrow& 8w(b) &< \bar{r}(\hat{x}(b), \theta^*_p) - \bar{r}(\hat{x}_{-}(b), \theta^*_p) && (\because \hat{x}(b) = x^*) \\
&\Rightarrow& 6w(b) &< \bar{r}(\hat{x}(b), \hat{\theta}(b)) - \bar{r}(\hat{x}_{-}(b), \hat{\theta}(b)) . && (\because E \text{ holds}) 
\end{align*}
This means that the If condition on line 15 will evaluate to true and the algorithm will stop and output an estimate $\hat{\Delta}$.

Now suppose the algorithm stops and does not output ``threshold exceeded'' which means that the If conditions on line 12 and line 15 were both true at some episode $b$. 
Let $x \neq \hat{x}(b)$ be arbitrary. We have the following chain of implications:
\begin{align*}
&& 6w(b) &< \bar{r}(\hat{x}(b), \hat{\theta}(b)) - \bar{r}(\hat{x}_{-}(b), \hat{\theta}(b)) && \text{(line 15)}\\
&\Rightarrow& 6w(b) &< \bar{r}(\hat{x}(b), \hat{\theta}(b)) - \bar{r}(x, \hat{\theta}(b)) && \text{($\hat{x}(b)$ unique maximizer by line 12)}\\
&\Rightarrow& 4w(b) &< \bar{r}(\hat{x}(b), \theta^*_p) - \bar{r}(x, \theta^*_p) . && (\because E \text{ holds}) 
\end{align*}
This means, along with Assumption 4, that $\hat{x}(b) = x^*$. We also have,
\begin{align*}
&& 6w(b) &< \bar{r}(\hat{x}(b), \hat{\theta}(b)) - \bar{r}(\hat{x}_{-}(b), \hat{\theta}(b)) && \text{(line 15)}\\
&\Rightarrow& 6w(b) &< \bar{r}(\hat{x}(b), \hat{\theta}(b)) - \bar{r}(x^*_{-}, \hat{\theta}(b)) && (\because \bar{r}(\hat{x}_{-}(b), \hat{\theta}(b)) \ge \bar{r}(x^{*}_{-}, \hat{\theta}(b))) \\
&\Rightarrow& 4w(b) &< \bar{r}(\hat{x}(b), \theta^*_p) - \bar{r}(x^*_{-}, \theta^*_p) && (\because E \text{ holds}) \\
&\Rightarrow& 4w(b) &< \bar{r}(x^*, \theta^*_p) - \bar{r}(x^*_{-}, \theta^*_p) && (\because \hat{x}(b) = x^*) \\
&\Rightarrow& 4w(b) &< \Delta . && \text{(def. of $\Delta$)} 
\end{align*}

Now we prove that the output $\hat{\Delta}$ lies in the right range. We have
\begin{align*}
\hat{\Delta} &= \bar{r}(\hat{x}(b), \hat{\theta}(b)) - \bar{r}(\hat{x}_{-}(b), \hat{\theta}(b)) && \text{(line 16)} \\
 &\ge \bar{r}(\hat{x}(b), \theta^*_p) - \bar{r}(\hat{x}_{-}(b), \theta^*_p) - 2w(b) && (\because E \text{ holds}) \\
 &= \bar{r}(x^*, \theta^*_p) - \bar{r}(\hat{x}_{-}(b), \theta^*_p) - 2w(b) && (\because \hat{x}(b) = x^*) \\
 &\ge \bar{r}(x^*, \theta^*_p) - \bar{r}(x^*_{-}, \theta^*_p) - 2w(b) && (\because \bar{r}(\hat{x}_{-}(b), \theta^*_p) \le \bar{r}(x^{*}_{-}, \theta^*_p)) \\
 &\ge \Delta - 2 w(b) && \text{(def. of $\Delta$)} \\
 &\ge \frac{\Delta}{2} . && (\because w(b) < \Delta/4) 
\end{align*}
Similarly,
\begin{align*}
\hat{\Delta} &= \bar{r}(\hat{x}(b), \hat{\theta}(b)) - \bar{r}(\hat{x}_{-}(b), \hat{\theta}(b)) && \text{(line 16)} \\
 &\le \bar{r}(\hat{x}(b), \hat{\theta}(b)) - \bar{r}(x^*_{-}, \hat{\theta}(b))  && (\because \bar{r}(\hat{x}_{-}(b), \hat{\theta}(b)) \ge \bar{r}(x^{*}_{-}, \hat{\theta}(b))) \\
 &= \bar{r}(x^*, \hat{\theta}(b)) - \bar{r}(x^*_{-},\hat{\theta}(b))  && (\because \hat{x}(b) = x^*) \\
 &\le \bar{r}(x^*, \theta^*_p) - \bar{r}(x^*_{-}, \theta^*_p) + 2w(b) && (\because E \text{ holds}) \\
 &\le \Delta + 2 w(b) && \text{(def. of $\Delta$)} \\
 &\le \frac{3\Delta}{2} . && (\because w(b) < \Delta/4) 
\end{align*}

\item
In this case $T_0 \le T_1(\delta)$ but it could still be that the algorithm stops not because the threshold is exceeded but because line 12 and line 15 were true
at some episode $b$. Clearly $b < T_0$, otherwise we would have output ``threshold exceeded'' and not produced an estimate $\hat{\Delta}$.
Under the event $E$, the previous part shows that if stopping occurs with an estimate $\hat{\Delta}$, it must be that $4w(b) < \Delta$, i.e.
\[
4  \sqrt{\dfrac{R^2 \beta_{\sigma}^2 \log (\frac{4e^2b^2}{\delta})}{b}}  < \Delta
\quad\Rightarrow
\quad
b > \frac{16 R^2 \beta_\sigma^2}{\Delta^2} \log \frac{4e^2}{\delta}  = T_2(\delta) .
\]
This means $T_0 > b > T_2(\delta)$.
\item
Finally, suppose Assumptions 1 through 3 hold but Assumption 4 fails. Event $E$ still holds with probability at least $1-\delta$. However, if there are
at least two optimal actions then, under $E$, their confidence intervals will always overlap and If condition on line 15 will never be true. That means that
the algorithm can only stop when the threshold $T_0$ is exceeded.
\end{enumerate}
\end{proof}

\subsubsection{Proof of Theorem~\ref{thm:combined}}

\begin{proof}
We break the proof into the two cases mentioned in the theorem statement.

{\bf Part 1: Assumption 4 fails.}
From Theorem \ref{thm:gap} we know, that with probability at least $1-\delta$, Algorithm~\ref{alg:gap} outputs ``threshold exceeded'' in this case.
Because of Eq.~\eqref{eq:boundshold}, we also have, for an optimal action $x^*$:
\[
|\bar{r}(\hat{x}(T_0), \theta^*_p) - \bar{r}(x^*,\theta^*_p)| \le 2 w(T_0)
\]
which implies a total regret of
\[
2 w(T_0) (T - T_0|\sigma|) \le 2 w(T_0) T
\]
in the remaining $T - T_0|\sigma|$ rounds since we execute line 5. The regret when Algorithm~\ref{alg:gap} was running is bounded by $R_{max} T_0|\sigma|$.
On the bad event, which occurs with probability at most $1-\delta$, the regret is at most $T R_{max}$ giving us a total expected regret of
\[
2w(T_0) T + T_0 |\sigma| R_{max} + \delta T R_{max}
=
2 \sqrt{\dfrac{R^2 \beta_{\sigma}^2 \log (\frac{4e^2T_0^2}{\delta})}{T_0}} T + T_0 |\sigma| R_{max} + \delta T R_{max}
\]
which is upper bounded by
\[
2(2 R \beta_\sigma |\sigma|^2 R^2_{max} T)^{2/3} \sqrt{\log (4e^2 T^3}) + R_{max}
\]
for $T_0 = \left(\frac{2R\beta_\sigma T}{|\sigma| R_{max}}\right)^{2/3}$ and $\delta = 1/T$.

{\bf Part 2: Assumption 4 holds.}
\emph{Case A:} $T_1(\delta) < T_0$. In this case, according to Theorem~\ref{thm:gap}, with probability at least $1-\delta$, Algorithm \ref{alg:gap} finishes in $T_1(\delta)$ episodes and outputs $0.5 \Delta \le \Delta \le 1.5 \Delta$.
This means $\Delta/36 R^2 \beta_\sigma^2 \le h \le \Delta^2/4 R^2 \beta_\sigma^2$. Therefore, by Theorem~\ref{thm:distributiondependentregret2}, we have, regret due to Algorithm~\ref{alg:PhaseForced}
is at most:
\[
\sum_{x \in \sigma} \Delta_x \frac{36 R^2 \beta_\sigma^2 \log T}{\Delta^2} + \frac{8 e^2 R^2 \beta_\sigma^2}{\Delta^2} .
\]
Overall, the expected regret is bounded by
\[
T_1(\delta) R_{max} |\sigma| + \sum_{x \in \sigma} \Delta_x \frac{36 R^2 \beta_\sigma^2 \log T}{\Delta^2} + \frac{8 e^2 R^2 \beta_\sigma^2}{\Delta^2} + R_{max} T \delta .
\]
For $\delta = 1/T$, this becomes
\[
\frac{256 R^2 \beta_\sigma^2}{\Delta^2} \log \frac{512 e^2 R^2 \beta_\sigma^2 T}{\Delta^2} R_{max} |\sigma|
+ \sum_{x \in \sigma} \Delta_x \frac{36 R^2 \beta_\sigma^2 \log T}{\Delta^2} + \frac{8 e^2 R^2 \beta_\sigma^2}{\Delta^2} + R_{max} .
\]

\emph{Case B:} $T_2(\delta) \le T_0 \le T_1(\delta)$. In this regime, Algorithm~\ref{alg:gap} can stop and output ``threshold exceeded'', in which case, expected regret is bounded, as in Part 1, by
\[
2(2 R \beta_\sigma |\sigma|^2 R^2_{max} T)^{2/3} \sqrt{\log (4e^2 T^3}) + R_{max} .
\]
However, it can also happen that Algorithm~\ref{alg:gap} stops and outputs $0.5 \Delta \le \Delta \le 1.5 \Delta$ with probability at least $1-\delta$. In that case, total expected regret is bounded, as in Part 2, Case A, by
\[
\frac{256 R^2 \beta_\sigma^2}{\Delta^2} \log \frac{512 e^2 R^2 \beta_\sigma^2 T}{\Delta^2} R_{max} |\sigma|
+ \sum_{x \in \sigma} \Delta_x \frac{36 R^2 \beta_\sigma^2 \log T}{\Delta^2} + \frac{8 e^2 R^2 \beta_\sigma^2}{\Delta^2} + R_{max} .
\]
Note that the above bound scales as $O(T_1(\delta))$ for $\delta = 1/T$, which is upper bounded by $O(T_2(\delta) \log T_2(\delta))$. But we know that $T_2(\delta) \le T_0$ which means the bound is no larger than
$O(T_0 \log T_0) = O(T^{2/3} \log T)$. So no matter what happens, regret is upper bounded by $O(T^{2/3} \log T)$ in this case.

\emph{Case C:} $T_0 \le T_2(\delta)$. With probability at least $1-\delta$, by Theorem~\ref{thm:gap}, in this case, Algorithm cannot stop and output $\hat{\Delta}$. Instead, it outputs ``threshold exceeded''.
When this happens, Algorithm~\ref{alg:PhaseForced} never gets called and only exploitation rounds follow (line 5). Regret is bounded, just as in Part 1, by
\[
2(2 R \beta_\sigma |\sigma|^2 R^2_{max} T)^{2/3} \sqrt{\log (4e^2 T^3}) + R_{max} .
\]

\end{proof}

\end{document}